\documentclass[11pt,a4paper]{article}

\usepackage{amsmath, amsthm, amsfonts, amssymb,color,geometry,enumerate, txfonts}			
\usepackage{float}
\usepackage{graphicx} 
\usepackage{caption} 
\usepackage{subcaption} 
\usepackage{fancybox}
\usepackage{cases}
\usepackage{fancyhdr}
\usepackage{bbm}
\usepackage{hyperref}
\usepackage{enumitem}

\newtheorem{theorem}{Theorem}[section]
\newtheorem{proposition}[theorem]{Proposition}
\newtheorem{lemma}[theorem]{Lemma}

\newtheorem{example}[theorem]{Example}

\makeatletter

\@addtoreset{equation}{section}
\makeatother

\newcommand{\Z}{\mathbb{Z}}    
\renewcommand{\epsilon}{\varepsilon}    

\newcommand{\lbeq}[1]{\label{eq:#1}}

\newcommand{\mP}{{\mathbb P}}

\newcommand{\vep}{\varepsilon}

\begin{document}
\title{Stability of energy landscape for Ising models}
\author{
Bruno Hideki Fukushima-Kimura\footnote{Faculty of Science, Hokkaido University,
Japan.} \hspace{2mm}Akira Sakai\footnotemark[1]\ \footnote{https://orcid.org/0000-0003-0943-7842} \hspace{2mm}  Hisayoshi Toyokawa\thanks{Institute of Mathematics for Industry, Kyushu University,
Japan.}  \hspace{2mm} Yuki Ueda\thanks{Department of Mathematics, Hokkaido University of Education, Japan.}
}
\date{}

\maketitle
\abstract{In this paper, we explore the stability of the energy landscape of an Ising Hamiltonian when subjected to two kinds of perturbations: a perturbation on 
the coupling coefficients and external fields,
and a perturbation on the underlying graph structure.
We give sufficient conditions so that the ground states of a given Hamiltonian are stable under perturbations of the first kind in terms of order preservation.
Here by order preservation we mean that the ordering of energy corresponding to two spin configurations in a perturbed Hamiltonian will be preserved in the original Hamiltonian up to a given error margin.
We also estimate the probability that the energy gap between ground states for the original Hamiltonian and the perturbed Hamiltonian is bounded by a given error margin when the coupling coefficients and local external magnetic fields of the original Hamiltonian are i.i.d. Gaussian random variables. In the end we show a concrete example of a system which is stable under perturbations of the second kind.}


\section{Introduction}\label{sec:intro}

Finding optimal solutions for combinatorial optimization problems, some of which are known to be NP-hard, is a very important problem.
Among many possible approaches to such problems, the application of Ising models to solve real social problems has been getting attention due to its versatility (see \cite{L14}).
More precisely, a given social combinatorial optimization problem can be mapped into a Hamiltonian $H$ on a graph $G=(V,E)$, whose expression is given by
\begin{align*}
H(\sigma)=-\sum_{b=\{x,y\}\in E}J_b\sigma_x\sigma_y-\sum_{x\in V}h_x\sigma_x
\end{align*}
for every Ising spin configuration $\sigma\in\{-1,1\}^V$, where $\{J_b\}_{b \in E}$ are coupling coefficients and $\{h_x\}_{x \in V}$ are local magnetic fields.
In that approach, an optimal solution for the intended combinatorial problem corresponds to a ground state (or global minimum) $\sigma_G$ of $H$, that is, $\sigma_G \in \mathrm{arg\, min}\,H$.
There are some well-known methods that can be applied to obtain a ground state.
Implementing a Markov chain Monte Carlo (such as Glauber dynamics and stochastic cellular automata) is known as a way to find an approximation for the Gibbs distribution whose highest peaks correspond to the ground states of $H$.
We refer for details to \cite{DSS12,H88,R13,HKKS19} and also \cite{HKKS21}.

However, as long as we use Ising machines or any computer to perform numerical simulations to find a ground state, we cannot avoid the error occurring due to the analog nature or the difficulty of representing real numbers (see \cite{AMH18}).
Because of these reasons, we should incorporate the error coming from the coupling coefficients and local magnetic fields by introducing a perturbed Hamiltonian.
Hence, our original Hamiltonian $H$ will be perturbed, originating a perturbed Hamiltonian $H_{\delta}$ whose coupling coefficients and local magnetic fields have a maximal error $\delta$.
Then, the following natural questions arise:
\begin{enumerate}[label=(\arabic*),start=1]
\item
For any pair of configurations which are ordered in terms of energy with respect to the perturbed Hamiltonian $H_\delta$, is that ordering preserved in the original Hamiltonian $H$, up to a given error margin? \label{question1}
\item
Given a Hamiltonian $H$ with coupling coefficients and local magnetic fields distributed as i.i.d. Gaussian random variables, what is the probability that the energy gap in $H$ between two ground states respectively for $H$ and the corresponding perturbed Hamiltonian $H_\delta$ is sufficiently small?\label{question2}
\end{enumerate}

In addition to the above questions \ref{question1} and \ref{question2}, the following problem is also important when using Ising machines and computers.
It may be somewhat a waste of resources taking all coupling coefficients and local magnetic fields into account.
It may be useful to ``eliminate" vertices of a given graph whose contribution to the total energy is relatively small, in order to save memories of computers.
Hence, we also have the following natural question:
\begin{enumerate}[label=(\arabic*),start=3]
\item
Can we find a subset of a given graph such that for an arbitrary choice of configuration outside of that region, the energy variation can be controlled? \label{question3}
\end{enumerate}
In this paper, we investigate the stability of energy landscape of a given Hamiltonian under perturbations from the view point of  order preservation, aiming at answering the questions we addressed above.
Thanks to the order preservation property, we can obtain better  estimates for the success probability of finding  a ground state  compared to the result given in \cite{AMH18}.

This paper is organized as follows.
In Section \ref{sec:settings},  we provide a precise formulation for the questions we just posed and raise them again.
In Section \ref{sec:solutions}, we answer the questions \ref{question1'} and \ref{question2'} from Section \ref{sec:settings}.
In Section \ref{pertgraph}, we provide an example together with a sufficient condition that guarantees a positive answer for question \ref{question3'}.


\section{Setting and the main questions}\label{sec:settings}
In this section, we introduce some necessary definitions and terminologies for discussing the stability of energy landscape.
Further we also introduce the notion of order preservation for a perturbed system, which plays a central role in this paper.
Here, order preservation means, roughly speaking, if we  take a ground state for a perturbed Hamiltonian (implemented by a device) then it should be close to the ground state for an original Hamiltonian (intended mathematical problem) in energy, up to a given error margin.

 Let us begin by introducing the precise setting.
Let $G=(V,E)$ be a finite simple graph with the vertex set $V$ and the edge set $E$.
The so-called {\it original Hamiltonian} $H$ with coupling coefficients $\{J_b\}_{b\in E}$ and external magnetic fields $\{h_x\}_{x\in V}$ on $G$ is defined by
\begin{align}
H(\sigma)=-\sum_{b=\{x,y\} \in E}J_b\sigma_x\sigma_y - \sum_{ x \in V}h_x\sigma_x
\end{align}
for each $\sigma=\{\sigma_x\}_{x\in V}\in\{-1,1\}^V$.
Such a function $H$ can be regarded as a cost function of an intended problem.
Given $\delta>0$, we denote by $H_\delta$ the {\it perturbed Hamiltonian} with the coupling coefficients $\{J'_b\}_{ b \in E}$ and external fields $\{h'_x\}_{x \in V}$, i.e.,
\begin{align}\lbeq{perturb}
H_\delta(\sigma)=-\sum_{ b=\{x,y\} \in E}J'_b\sigma_x\sigma_y-\sum_{ x \in V}h'_x
 \sigma_x
\end{align}
where the $J'_b$'s and $h'_x$'s satisfy the bounds $\sup_b\lvert J_b-J'_b\rvert\le\delta$ and $\sup_x\lvert h_x-h'_x\rvert\le\delta$.
This perturbation will be often interpreted as a round-off  in the following way.
Let $(J_b^{(1)}J_b^{(2)}\dots)$ and $(h_x^{(1)}h_x^{(2)}\dots)$ be the binary expansions of  the fractional parts of $J_b$ and $h_x$, i.e.,
\begin{align}
J_b=J_b^{(0)}+\sum_{i\ge1}\frac{J_b^{(i)}}{2^i},\quad h_x=h_x^{(0)}+\sum_{i\ge1}\frac{h_x^{(i)}}{2^i}
\end{align}
where $J_b^{(0)},h_x^{(0)}\in\Z$ and $J_b^{(i)},h_x^{(i)}\in\{0,1\}$ for $i\ge1$. If we set $J_b'=J_b^{(0)}+\sum_{i=1}^{N}\frac{J_b^{(i)}}{2^i}$ and $ h_x'=h_x^{(0)}+\sum_{i=1}^N\frac{h_x^{(i)}}{2^i}$ in the equation (\ref{eq:perturb}), then the error $\delta$ can be taken as $2^{-N}$.
It means that the perturbed Hamiltonian $H_{\delta}$ is obtained by  rounding off the given parameters $J_b$'s and $h_x$'s uniformly from  the $(N+1)$-th digit of their binary expansions.

The main purpose of this paper is to clarify the stability of the ground states for a given Hamiltonian under a perturbation in terms of order preservation.
In this paper, we will answer the following questions:
\begin{enumerate}[label=(\arabic*'),start=1]
\item
Find a $\delta>0$ corresponding to a given $\vep>0$, so that, for any pair $(\sigma,\tau)$ that satisfies 
$H_{\delta}(\sigma)\ge H_{\delta}(\tau)$, the ordering is preserved in $H$
up to the error margin $\vep\sup_{\xi,\eta}\left\lvert H(\xi)-H(\eta)\right\rvert$, 
i.e., 
\begin{align}
H(\sigma)\ge H(\tau)-\vep\sup_{\xi,\eta}\left\lvert H(\xi)-H(\eta)\right\rvert.
\end{align}
Here, $\sup_{\xi,\eta}\left\lvert H(\xi)-H(\eta)\right\rvert$ is the total margin of the original Hamiltonian.  \label{question1'}
\item \label{question2'}
Let $(\Omega, \mathcal{F},\mathbb{P})$ be a probability space and let $\{J_b\}_{b\in E}$ and $\{h_x\}_{x\in V}$ be mutually independent standard Gaussian random variables on this probability space.
Estimate the probability that the energy gap in $H$ between ground states for $H$ and $H_{\delta}$, say $\sigma_G$ and $\tilde{\sigma}_G$, respectively, is controled by the given error margin, explicitly,
\begin{align}\label{eq:quest2}
\mathbb{P}\left(0 \leq H(\tilde{\sigma}_G)-H(\sigma_G)\le\vep\sup_{\xi,\eta}\left\lvert H(\xi)-H(\eta)\right\rvert\right).
\end{align}  
\end{enumerate}

A different aspect of stability of a given system is to find a nontrivial subsystem so that the energy gap between any two spin configurations whose spins restricted to the subgraph coincide is bounded above by a given error margin.
Also, at the same time, we require that the number of vertices that can be disregarded is at least of order $N^{\alpha}$, where $N=\lvert V\rvert$ and $\alpha\in {[}0,1)$, so that such a number can go to infinity as $N\to\infty$.
In the later part of this paper, we  answer the following question for a particular case:
\begin{enumerate}[label=(\arabic*'),start=3]
\item 
Let $\lvert V\rvert =N$, and let $\{J_b\}_{b\in E}$ and $\{h_x\}_{x\in V}$ be mutually independent standard Gaussian random variables.
Find a subset $V_0\subset V$ for a given $\vep>0$ and $\alpha\in {[}0,1)$ such that
\begin{align*}
\mP\left(\sup_{\sigma,\tau\in\{-1,1\}^N} \left\lvert H(\sigma)-H(\sigma_{V_0}, \tau_{V\setminus V_0}) \right\rvert < \vep\sup_{\xi,\eta}\left\lvert H(\xi)-H(\eta)\right\rvert  \;\&\; CN^{\alpha} \leq \lvert V\setminus V_0\rvert < N \right)
\end{align*}
is close to $1$, where $\sigma_{V_0}\in\{-1,1\}^{V_0}$ is the spin configuration $\sigma$ restricted to $V_0$ and $\tau_{V\setminus V_0}\in\{-1,1\}^{V\setminus V_0}$ is the restriction of the spin configuration $\tau$ to $V\setminus V_0$. \label{question3'}
\end{enumerate}

Questions \ref{question1'}, \ref{question2'} and \ref{question3'} above correspond to questions \ref{question1}, \ref{question2} and \ref{question3} from Section \ref{sec:intro}, respectively.
In Section \ref{sec:solutions}, we investigate the first two questions above, where for the second one we adopt two different approaches. We obtain answers for question \ref{question2'} by means of a method involving the $L^\infty$-distance and a graph's structure approach, and we compare these two methods for three different graphs.
Specifically, we consider sufficient conditions on the perturbation $\delta$ to satisfy order preservation, and calculate the probability that such a sufficient condition holds.
In Section \ref{pertgraph}, we obtain an answer for the question \ref{question3'} when the graph is a one-dimensional torus $\Z/N\Z$ without external fields.


\section{Stability under  a Hamiltonian perturbation} \label{sec:solutions}
This part is dedicated to provide solutions for questions \ref{question1'} and \ref{question2'} just posed in the end of the previous section. Before we proceed to the next sections, let us introduce the quantity $R_{H}$  defined by
\begin{equation}
R_{H}\coloneqq\max_{\xi,\eta}\lvert H(\xi)-H(\eta)\rvert,
\end{equation}
which is defined whenever a Hamiltonian $H$ is given. Moreover, if $G = (V,E)$ is a finite simple graph, then we define $k_{G}$ by
\begin{equation}
	k_{G} := |E| + |V|.
\end{equation}

Keeping in mind the mathematical setting introduced in the beginning of Section \ref{sec:settings}, let us start by showing that the order preservation property holds, that is,
let us first answer the question \ref{question1'}, which consists in finding a  $\delta>0$ corresponding to a given $\epsilon>0$ such that $H_\delta(\sigma)\ge H_\delta(\tau)$ implies $H(\sigma)\ge H(\tau)-\epsilon R_{H}$; and later on, assuming some randomness on the spin-spin couplings and external fields, we adopt two different approaches to answer question \ref{question2'} and estimate the probability that the condition $H(\tilde{\sigma}_G)- H(\sigma_G)\le \epsilon R_{H}$ is satisfied. 

In order to solve the second problem, we will adopt two distinct  approaches: a method that relies on uniform estimates and a method where combinatorial estimates are considered instead, which will be presented in Sections \ref{sec:FA} and \ref{sec:SA}, respectively.
In the last part of this section, we compare these two methods and conclude that depending on the underlying graph structure of the problem, one of them will give us a better lower bound for the probability from equation (\ref{eq:quest2}).


\subsection{Order preservation of energy}

The answer for question \ref{question1'} from Section \ref{sec:settings} is provided by Theorem \ref{thm:order_preservation}, however, let us show first a preliminary result.

In \cite{HKKS19}, we have already established lower bounds for the total margin $R_{H}$ of the Hamiltonian $H$, but for the reader's convenience we include  its proof in the present paper.

\begin{lemma}[See \cite{HKKS19}]\label{lem:total margin}
Let us consider a finite simple graph $G=(V,E)$ and a Hamiltonian $H$ written in the form
\begin{align*}
H(\sigma)=-\sum_{b=\{x,y\}\in E} J_b \sigma_x\sigma_y -\sum_{x\in V}h_x \sigma_x
\end{align*}
for each $\sigma\in \{-1,1\}^V$. Then, we have
\begin{align*}
R_{H}\ge \sqrt{v_H}, \qquad \text{where } v_H\coloneqq\sum_b J_b^2+\sum_x h_x^2.
\end{align*}
\end{lemma}

\begin{proof}
For any probability measure $\mu$ on the configuration space $\{-1,1\}^V$, we have
\begin{align*}
R_{H}\ge  \left(\mathbb{E}_\mu[H^2]-\mathbb{E}_\mu[H]^2\right)^{1/2},
\end{align*}
where $\mathbb{E}_\mu$ stands for the expectation with respect to the probability measure $\mu$.
If $\mu$ is particularly chosen as the uniform distribution on $\{-1,1\}^V$, then we have

\begin{align*}
\mathbb{E}_\mu[H]&\coloneqq\frac{1}{2^{\lvert V\rvert}} \sum_\sigma H(\sigma)=0,{}
\end{align*}
and
\begin{align*}
\mathbb{E}_\mu[H^2]&\coloneqq\frac{1}{2^{\lvert V\rvert}} \sum_\sigma H(\sigma)^2\\
&=\frac{1}{2^{\lvert V\rvert}} \sum_\sigma \left( -\sum_{b=\{x,y\}\in E} J_b\sigma_x\sigma_y-\sum_{x\in V}h_x\sigma_x\right)^2 \\
&=\frac{1}{2^{\lvert V\rvert}} \sum_\sigma \left(\sum_{b,b'\in E} J_bJ_{b'}\sigma_x\sigma_y\sigma_{x'}\sigma_{y'}+\sum_{x,x'\in V}h_xh_{x'}\sigma_x\sigma_{x'} \right)\\
&=\frac{1}{2^{\lvert V\rvert}} \sum_\sigma \left( \sum_{b\in E}J_b^2+\sum_{x\in V}h_x^2\right) ={v_H}.
\end{align*}
Therefore, $R_{H}\ge  \left(\mathbb{E}_\mu[H^2]-\mathbb{E}_\mu[H]^2\right)^{1/2} =\sqrt{v_H}$.
\end{proof}

In order to prove the next result, it is convenient to consider the following notation introduced by \cite{AMH18}.
For any Ising spin configurations $\sigma$ and $\tau$, we consider the sets $D_{\sigma,\tau} $ and $W_{\sigma,\tau} $ defined by
\[D_{\sigma,\tau} \coloneqq\{x\in V: \sigma_x\tau_x =-1\}\] 
and
\[W_{\sigma,\tau} \coloneqq\{\{x,y\}\in E: \sigma_x\sigma_y\tau_x\tau_y =-1\},\]
where the products $\sigma_x\tau_x$ and $\sigma_x\sigma_y\tau_x\tau_y$ are called the {\it spin overlap} and the {\it link overlap}, respectively.

\begin{theorem}\label{thm:order_preservation}
Given $\epsilon>0$ and configurations $\sigma$ and $\tau$, if the condition
\begin{align*}
0 < \delta k_{G} \le \frac{1}{2} \epsilon \sqrt{v_{H}}
\end{align*}
is satisfied, then $H_\delta(\sigma)\ge H_\delta(\tau)$ implies $H(\sigma)\ge H(\tau)-\epsilon R_{H}$.
\end{theorem}
\begin{proof}
If we suppose that $H_\delta(\sigma)\ge H_\delta(\tau)$, then, we have
\begin{align*}
H(\tau)-H(\sigma)& = \left(H(\tau)-H_\delta(\tau)\right) + H_\delta(\tau)-H_\delta(\sigma) + \left(H_\delta(\sigma)-H(\sigma)\right)\\
&\le \left(H_\delta(\sigma)-H(\sigma)\right) - \left(H_\delta(\tau)- H(\tau)\right)\\
&\le \sum_{b=\{x,y\}\in E} \left\lvert J_b-J_b'\right\rvert \left\lvert\sigma_x\sigma_y-\tau_x\tau_y\right\rvert+ \sum_{x\in V}\left\lvert h_x-h_x'\right\rvert \left\lvert\sigma_x-\tau_x\right\rvert\\
&=2\sum_{b\in W_{\sigma,\tau}}\left\lvert J_b-J_b'\right\rvert + 2\sum_{x\in D_{\sigma,\tau}}\left\lvert h_x-h_x'\right\rvert\\
&\le 2\delta (\lvert W_{\sigma,\tau}\rvert+\lvert D_{\sigma,\tau}\rvert).
\end{align*}
Since $\lvert W_{\sigma,\tau}\rvert \leq |E|$, $\vert D_{\sigma,\tau}\rvert \leq |V|$, $k_{G} = |E| + |V|$, and  $R_{H}\ge \sqrt{v_{H}}$, then, by our assumption, we obtain
\begin{align*}
H(\tau)-H(\sigma)\le 2\delta k_{G} \le \epsilon \sqrt{v_{H}} \le \epsilon R_{H}.
\end{align*}
Therefore, the conclusion of this theorem  follows.
\end{proof}

\subsection{Stability of ground states: first approach}\label{sec:FA}
In the previous subsection, we did not assume any randomness on the spin-spin couplings $J_b$'s and local external fields $h_x$'s.
In this subsection, let us consider the same  setting as stated in question \ref{question2'} from Section \ref{sec:settings}.
Precisely speaking, we assume that $\{J_b\}_{b \in E}$ and $\{h_x\}_{x \in V}$ are mutually independent random variables distributed according to a standard Gaussian distribution. 

Under such assumptions, let us estimate the probability that  the inequality $H(\tilde{\sigma}_G)-H(\sigma_G)\le \epsilon R_{H}$  holds,
 where $\epsilon$ is a given positive constant, by using a method that relies on uniform bounds with respect to certain spin configurations.
In the following lemma, we provide an upper bound for the difference $H(\tilde{\sigma}_G)-H(\sigma_G)$.

\begin{lemma}
Given $\delta>0$, if $\sigma_{G}$ and $\tilde{\sigma}_G$ are ground states for $H$ and $H_{\delta}$, respectively, then, we have 
\begin{equation}
H(\tilde{\sigma}_G)-H(\sigma_G)\le 2\delta k_{G}.
\end{equation}
\end{lemma}
\begin{proof}
It follows from the definition of a ground state that $H_\delta(\tilde{\sigma}_G)-H_\delta(\sigma_G)\le 0$, then
\begin{align*}
H(\tilde{\sigma}_G)-H(\sigma_G)&= \left(H(\tilde{\sigma}_G)-H_\delta(\tilde{\sigma}_G)\right) + H_\delta(\tilde{\sigma}_G)-H_\delta(\sigma_G) + \left(H_\delta(\sigma_G)-H(\sigma_G)\right)\\
&\le 2\lVert H_\delta-H\rVert_\infty,
\end{align*}
where $\|\cdot\|_{\infty}$ stands for the uniform norm, as usual. Furthermore, for any spin configuration $\sigma$, we have
\begin{align*}
\lvert H_{\delta}(\sigma)-H(\sigma)\rvert&=\left\lvert\sum_{b=\{x,y\}\in E} (J_b-J_b')\sigma_x\sigma_y+\sum_{x\in V}(h_x-h_x') \sigma_x\right\rvert\\
&\le \sum_{b\in E}\left\lvert J_b-J_b' \right\lvert+\sum_{x\in V} \left\rvert h_x-h_x'\right\rvert\\
&\le \delta (\lvert E\rvert+\lvert V\rvert)=\delta k_G.
\end{align*}
Then, $\lVert H_\delta-H\rVert_\infty \le \delta k_G$, therefore, we conclude the proof.
\end{proof}

By the lemma  above , it follows that
\begin{align*}\label{ineq:total}
\mathbb{P}\left(H(\tilde{\sigma}_G)-H(\sigma_G)\le \epsilon R_{H}\right) \ge \mathbb{P}\left( \delta k_{G} \le \frac{1}{2}\epsilon R_{H}\right),
\end{align*}
and by using the fact that $R_{H}\ge \sqrt{v_H}$ (see Lemma \ref{lem:total margin}), we conclude that
\begin{equation}
\mathbb{P}\left(H(\tilde{\sigma}_G)-H(\sigma_G)\le \epsilon R_{H}\right) \ge \mathbb{P}\left( \delta k_{G} \le \frac{1}{2} \epsilon \sqrt{v_{H}}\right).
\end{equation}
Finally, we have the following estimation for the probability that $H(\tilde{\sigma}_G)-H(\sigma_G)\le \epsilon R_{H}$ holds, which consists of one of the answers for the question \ref{question2'}.

\begin{theorem}\label{thmLinfty}
Let $\{J_b\}_{b \in E}$ and $\{h_x\}_{x \in V}$ be mutually independent standard Gaussian random variables.  It follows that
\begin{equation}
\mathbb{P} \left( H(\tilde{\sigma}_G)-H(\sigma_G)\le \epsilon R_{H} \right) \ge 1- \gamma\left(k_G; \left(\frac{2\delta k_G}{\epsilon}\right)^2\right),
\end{equation}
where $\gamma(s;x)$ is the distribution function of the chi-square distribution with $s>0$ degrees  of freedom, that is, 
\begin{align*}
\gamma(s;x)\coloneqq\frac{1}{2^{s/2}\Gamma(s/2)}\int_0^x t^{s/2-1}e^{-t/2} dt
\end{align*}
for $x \geq 0$, and $\gamma(s;x) \coloneqq 0$ for $x<0$.
\end{theorem}
\begin{proof}
It follows from the above discussion that we have
\begin{align*}
\mathbb{P}\left(H(\tilde{\sigma}_G)-H(\sigma_G)\le \epsilon R_{H}\right) &\ge \mathbb{P}\left( \delta k_{G} \le \frac{1}{2} \epsilon \sqrt{v_{H}}\right)\\
&=\mathbb{P}\left( v_{H} \ge \left(\frac{2\delta k_G}{\epsilon} \right)^2\right)\\
&=1-\mathbb{P}\left( v_{H} < \left(\frac{2\delta k_G}{\epsilon} \right)^2\right).
\end{align*}
Since $\{J_b\}_{b \in E}$ and $\{h_x\}_{x \in V}$ are mutually independent random variables distributed according to a standard Gaussian distribution, then the random variable $v_{H}$ is distributed as the chi-square distribution with $k_{G}$ degrees of freedom.
Therefore,
\begin{align*}
\mathbb{P}\left( v_{H} < \left(\frac{2\delta k_G}{\epsilon} \right)^2\right)=\gamma\left(k_G; \left(\frac{2\delta k_G}{\epsilon}\right)^2\right).
\end{align*}
Thus, we obtain the lower bound of the target probability.
\end{proof}


\subsection{Stability of ground states: second approach}\label{sec:SA}

Before we proceed, let us point out the fundamental difference between the uniform approach and the current approach to solve question \ref{question2'}. Note that, 
if we use the same computations as considered in the proof of Theorem \ref{thm:order_preservation} in the particular case where $\tau = \tilde{\sigma}_G$ and $\sigma = \sigma_{G}$ and use the fact that
$H_\delta(\sigma_{G})\ge  H_\delta(\tilde{\sigma}_G)$, then it follows  that
\begin{equation}\label{eq:OPbound}
	H(\tilde{\sigma}_G) - H(\sigma_G) \le 2\delta(\lvert W_{\sigma_G,\tilde{\sigma}_G}\rvert+\lvert D_{\sigma_G,\tilde{\sigma}_G}\rvert).
\end{equation}
Recall that the proof of Theorem \ref{thmLinfty} fundamentally relied on the fact that, by using the $L^{\infty}$-distance estimates, the left-hand side of equation (\ref{eq:OPbound}) could be bounded above by $2\delta k_{G}$. Note that the right-hand side of equation (\ref{eq:OPbound}) is also bounded above by $2\delta k_{G}$, therefore, let us explore the geometry of the underlying graph 
$G$ in order to see whether it is possible to obtain better bounds.

The value of $\lvert W_{\sigma_G,\tilde{\sigma}_G}\rvert+\lvert D_{\sigma_G,\tilde{\sigma}_G}\rvert$ depends on the underlying graph structure and the relationship between the ground states $\sigma_G$ and $\tilde{\sigma}_G$. Therefore, we should check the value of $\lvert W_{\sigma_G,\tilde{\sigma}_G}\rvert+\lvert D_{\sigma_G,\tilde{\sigma}_G}\rvert$ for the intended problem.
In general, we look for a uniform estimation for the value of $\lvert W_{\sigma,\tau}\rvert+\lvert D_{\sigma,\tau}\rvert$ for any $\sigma$ and $\tau$ since the ground states $\sigma_G$ and $\tilde{\sigma}_G$ in practice are unknown.
First, let us show the following lemma.

\begin{lemma}\label{lem2}
For any two configurations $\sigma$ and $\tau$, we have
\begin{align*}
\lvert W_{\sigma,\tau}\rvert \leq (\deg G) \cdot \min \left\{\lvert D_{\sigma,\tau}\rvert, \lvert V\setminus D_{\sigma,\tau}\rvert\right\},
\end{align*}
where $\deg G$ stands for the maximum degree of $G$.
\end{lemma}

\begin{proof}
Let us assume that $\lvert D_{\sigma,\tau}\rvert =s$, for some $s$ such that $0\leq s \leq \lvert V\rvert$.
Then, let us enumerate $D_{\sigma,\tau}$ as $D_{\sigma,\tau}=\{x_1, \dots , x_s\}\subset V$, where $x_i\in V$ for each $i=1,\dots, s$.
Moreover, we have $\lvert V\setminus D_{\sigma,\tau}\rvert =\lvert V\rvert -s$, and therefore we can write $V\setminus D_{\sigma,\tau}=\{y_1, \dots, y_{\lvert V\rvert -s}\}\subset V$, where $y_i\in V$ for each $i=1,\dots ,\lvert V\rvert-s$.
By the definition of $D_{\sigma,\tau}$, we have $\sigma_{x_i} \tau_{x_i} =-1$ for every $i=1, \dots, s$ and $\sigma_{y_j} \tau_{y_j}=1$ for all $j=1, \dots, \lvert V\rvert -s$.
If $\{x_i,x_j\}\in E$ for distinct $i$ and $j$ in $\{1, \dots, s\}$, then
\begin{align*}
\sigma_{x_i}\sigma_{x_j} \tau_{x_i} \tau_{x_j}  = (\sigma_{x_i} \tau_{x_i}) (\sigma_{x_j} \tau_{x_j})=(-1)^2=1.
\end{align*}
Thus, $\{x_i,x_j\}\notin W_{\sigma,\tau}$.
In a similar way, we conclude that in case $\{y_i,y_j\}\in E$ for distinct $i$ and $j$ in $\{1,\dots ,\lvert V\rvert -s\}$, it follows that $\{y_i,y_j\}\notin W_{\sigma,\tau}$.
If $\{x_i,y_j\}\in E$ for some $i\in\{1,\dots , s\}$ and some $j\in \{1,\cdots , \lvert V\rvert -s\}$, then 
\begin{align*}
\sigma_{x_i}\sigma_{y_j} \tau_{x_i}\tau_{y_j} = (\sigma_{x_i}\tau_{x_i}) (\sigma_{y_j} \tau_{y_j}) =(-1)\times 1=-1.
\end{align*}
Hence, $\{x_i,y_j\}\in W_{\sigma,\tau}$.
It follows that
\begin{align*}
W_{\sigma,\tau}=\{\{x,y\}\in E: x=x_i, y=y_j \text{ for some } i,j\}.
\end{align*}
Therefore, we have
\begin{align*}
\lvert W_{\sigma,\tau}\rvert\leq (\deg G) \min \{ \lvert D_{\sigma,\tau}\rvert ,\lvert V\setminus D_{\sigma,\tau}\rvert\}.
\end{align*}
\end{proof}

\begin{proposition}\label{prop:unif_W+D}
For any graph $G$, let $\sigma$ and $\tau$ be two spin configurations. Then, we have
\begin{align}
\lvert W_{\sigma,\tau}\rvert +\lvert D_{\sigma,\tau}\rvert\le \frac{(\deg G+1)\lvert V\rvert}{2}.
\end{align}
\end{proposition}

\begin{proof}
Using Lemma \ref{lem2}, if $\lvert D_{\sigma,\tau}\rvert\leq \lvert V\rvert/2$ then
\begin{align*}
\lvert D_{\sigma,\tau}\rvert +\lvert W_{\sigma,\tau}\rvert\leq (\deg G+1)\lvert D_{\sigma,\tau}\rvert\leq \frac{\deg G+1}{2}\lvert V\rvert,
\end{align*}
otherwise, if $\lvert D_{\sigma,\tau}\rvert > \lvert V\rvert/2$, it follows that
\begin{align*}
\lvert D_{\sigma,\tau}\rvert +\lvert W_{\sigma,\tau}\rvert &\leq (\deg G) \lvert V\setminus D_{\sigma,\tau}\rvert +\lvert D_{\sigma,\tau}\rvert\\
&=(\deg G)\lvert V\rvert -(\deg G-1)\lvert D_{\sigma,\tau}\rvert\\
&\leq (\deg G)\lvert V\rvert -\frac{\deg G-1}{2}\lvert V\rvert\\
&=\frac{\deg G+1}{2}\lvert V\rvert.
\end{align*}
\end{proof}

Thus, we have the following theorem which is another answer for the question \ref{question2'} (see Theorem \ref{thmLinfty} for an alternative approach to the question \ref{question2'}).

\begin{theorem}\label{thmop}
Let $\{J_b\}_{b \in E}$ and $\{h_x\}_{x \in V}$ be mutually independent standard Gaussian random variables. Then, we have 
\begin{equation}
\mathbb{P}\left(H(\tilde{\sigma}_G)-H(\sigma_G)\le \epsilon R_{H}\right)\ge 1-\gamma\left( k_G; \left(\frac{\delta\lvert V\rvert(\deg G+1)}{\epsilon}\right)^2\right).
\end{equation}
\end{theorem}

\begin{proof}
Analogously as in the proof of Theorem \ref{thmLinfty}, it follows from equation (\ref{eq:OPbound}), $R_{H}\ge \sqrt{v_H}$ and Proposition \ref{prop:unif_W+D} that 
\begin{align*}
\mathbb{P}\left(H(\tilde{\sigma}_G)-H(\sigma_G)\le \epsilon R_{H}\right){}
&\ge \mathbb{P}\left(\delta\le \frac{\epsilon \sqrt{v_{H}}}{2(\lvert W_{\sigma_G,\tilde{\sigma}_G}\rvert +\lvert D_{\sigma_G,\tilde{\sigma}_G}\rvert )} \right)\\
&=\mathbb{P} \left(  v_{H}\ge \left(\frac{2\delta(\lvert W_{\sigma_G,\tilde{\sigma}_G}\rvert +\lvert D_{\sigma_G,\tilde{\sigma}_G}\rvert)}{\epsilon}\right)^2\right)\\
&\ge \mathbb{P}\left( v_{H} \ge \left(\frac{\delta\lvert V\rvert(\deg G+1)}{\epsilon}\right)^2\right)\\
&=1-\gamma\left( k_G; \left(\frac{\delta\lvert V\rvert(\deg G+1)}{\epsilon}\right)^2\right),
\end{align*}
where we used the fact that $v_{H}$ is distributed  according to a chi-square distribution with $k_{G}$ degrees of freedom.
\end{proof}




\subsection{Comparison between approaches}

In the rest of this section, we compare the methods presented in Sections \ref{sec:FA} and \ref{sec:SA}
passing through several examples to which we apply Proposition \ref{prop:unif_W+D}.

The first example is the case where we consider complete graphs including the SK model.
If we consider complete graphs, then Theorem \ref{thmop} provides us with better results if compared to Theorem \ref{thmLinfty}.

\begin{example}\label{ex:completegraph}
If $G$ is a complete graph (that is, all vertices are connected to each other) with $N$ vertices, then we have
\begin{align*}
\frac{\deg G+1}{2}\lvert V\rvert=\frac{N^2}{2}.
\end{align*}
On the other hand, the value of $k_{G}$ will be given by
\begin{align*}
k_{G} :=\lvert E\rvert+\lvert V\rvert =\frac{N(N-1)}{2}+N=\frac{N(N+1)}{2}.
\end{align*}
Therefore, 
\begin{align*}
\frac{\deg G+1}{2}\lvert V\rvert< k_{G}.
\end{align*}
Hence the uniform upper bound for $\lvert W_{\sigma,\tau}\rvert+\lvert D_{\sigma,\tau}\rvert$ we obtained in Proposition \ref{prop:unif_W+D} is always better than $k_G$.
Furthermore, we can calculate the explicit value of $\lvert W_{\sigma,\tau}\rvert+\lvert D_{\sigma,\tau}\rvert$ when $G$ is a complete graph.
From the proof of Lemma \ref{lem2}, by assuming that $G$ is a complete graph, 
we can say that $\lvert W_{\sigma,\tau}\rvert =\lvert D_{\sigma,\tau}\rvert (\lvert V\rvert -\lvert D_{\sigma,\tau}\rvert)$. Therefore,
\begin{align*}
\lvert W_{\sigma,\tau}\rvert +\lvert D_{\sigma,\tau}\rvert = \lvert D_{\sigma,\tau}\rvert (N + 1 - \lvert D_{\sigma,\tau}\rvert)
\le \frac{(N+1)^2}{4},
\end{align*}
and the proof of Theorem \ref{thmop} implies
\begin{align*}
\mathbb{P}\left(H(\tilde{\sigma}_G)-H(\sigma_G)\le \epsilon R_{H}\right)
&\geq \mathbb{P} \left(  v_{H}\ge \left(\frac{2\delta(\lvert W_{\sigma_G,\tilde{\sigma}_G}\rvert +\lvert D_{\sigma_G,\tilde{\sigma}_G}\rvert)}{\epsilon}\right)^2\right)\\
&\ge \mP\left( v_{H}\ge\frac{\delta^{2} (N+1)^4}{4\vep^2} \right)\\
&= 1-\gamma\left( k_G; \frac{\delta^{2} (N+1)^4}{4\vep^2}\right).
\end{align*}
\end{example}

The following example considers King's graphs and Theorem \ref{thmop} works better than Theorem \ref{thmLinfty} as well as the above example.

\begin{example}
Let $G$ be an $N\times M$ King's graph.
The $N\times M$ King's graph can be visualized as an $N\times M$ chessboard where each of its squares corresponds to a vertex of the graph, and each edge represents
a legal move of a king in a chess game. In that way, the inner vertices of the graph have $8$ neighbors each, while the vertices in the corners have $3$ neighbors each, and each of the remaining vertices on the
sides of the graph has $5$ neighbors. 
For an $N\times M$ King's graph, we have
\begin{align*}
\frac{\deg G+1}{2}\lvert V\rvert=\frac{9}{2}MN,
\end{align*}
since $\deg G=8$. Moreover, we have
\begin{align*}
k_{G} = \lvert E\rvert+\lvert V\rvert=5MN-3(M+N)+2.
\end{align*}
If $M$ and $N$ are sufficiently large, then we have
\begin{align*}
\frac{\deg G+1}{2}\lvert V\rvert< k_{G}.
\end{align*}
\end{example}

In the following example, differently from the previous ones, we can see that the estimate provided by Theorem \ref{thmLinfty} suits better than that of Theorem \ref{thmop}.

\begin{example}
If $G$ is a star graph with degree $k\ge 3$, that is, $G$ consists of one vertex placed in the center and other $k$ vertices connected only with the center, then 
\begin{align*}
\frac{\deg G+1}{2}\lvert V\rvert=\frac{(k+1)^2}{2}.
\end{align*}
Furthermore, we have
\begin{align*}
k_{G} = \lvert E\rvert+\lvert V\rvert=2k+1.
\end{align*}
Therefore, we obtain
\begin{align*}
\frac{\deg G+1}{2}\lvert V\rvert> k_{G}.
\end{align*}
\end{example}

\begin{figure}[H]
        \centering
        \begin{subfigure}[b]{0.44\textwidth}
                \includegraphics[width=\textwidth]{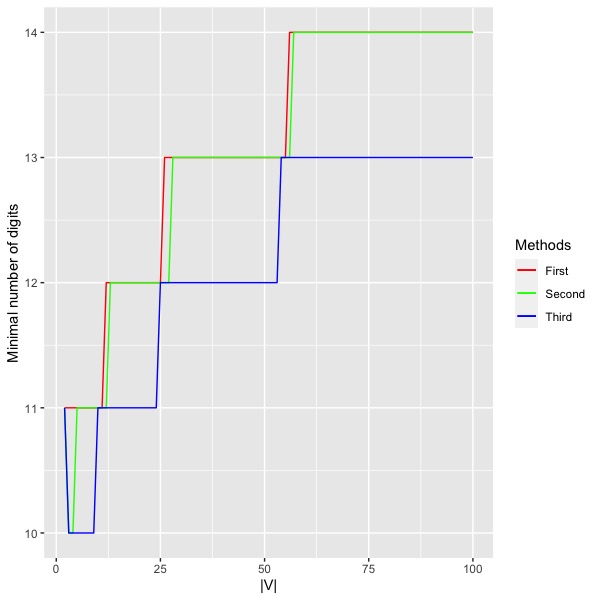}
                \caption{Complete graph.}
                \label{fig:dig1}
        \end{subfigure}       
        \begin{subfigure}[b]{0.44\textwidth}
                \includegraphics[width=\textwidth]{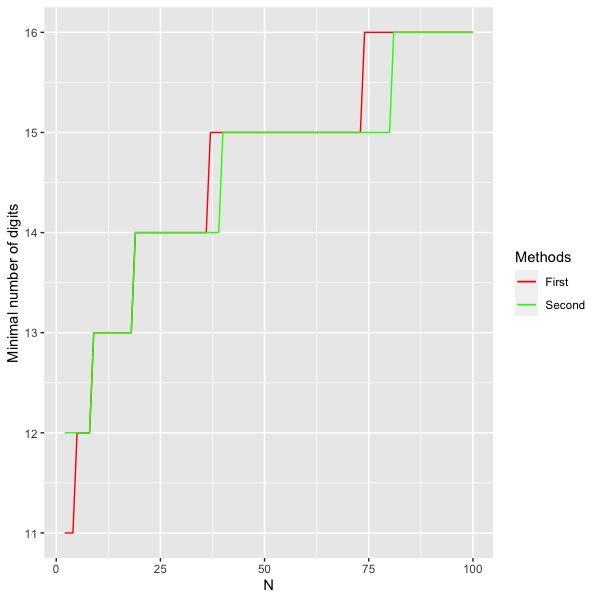}
                \caption{$N \times N$ King's graph.}
                \label{fig:dig2}
        \end{subfigure}
        \begin{subfigure}[b]{0.44\textwidth}
                \includegraphics[width=\textwidth]{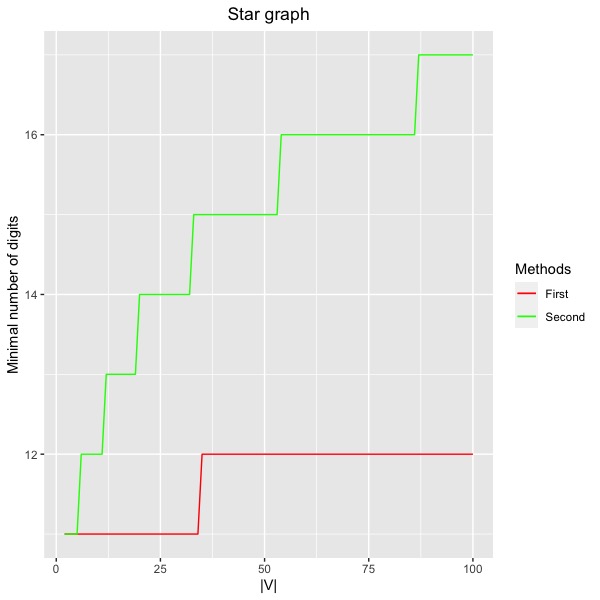}
                \caption{Star graph.}
        \end{subfigure}       
    \caption{Minimal number of digits to be considered in the binary expansions of the parameters so that with probability higher than $99\%$ the difference
    $H(\tilde{\sigma}_G) - H(\sigma_G)$ represents a value smaller than $1\%$ of $R_H$, as a function of the size of the graph.}
    \label{fig:digit}
\end{figure}
According to the above examples, we conclude that it is not always possible to guarantee that the uniform upper bound of $\lvert W_{\sigma,\tau}\rvert +\lvert D_{\sigma,\tau}\rvert$ provided by Proposition \ref{prop:unif_W+D} works better than $k_G=\lvert E\rvert +\lvert V\rvert$.
Thus, we may have to consider such bounds separately when considering different graphs in order to obtain an optimal estimate for the probability that inequality $H(\tilde{\sigma}_G)-H(\sigma_G)<\epsilon R_{H}$ holds.  

Let us consider again the problem of stability where we take into account only a finite number of terms in the binary expansions of the parameters $(J_b)_{b \in E}$ and $(h_x)_{x \in V}$ as we illustrated in the beginning of Section \ref{sec:settings}. In Figure \ref{fig:digit}, corresponding to the sizes  of different graphs, we show the minimum number of digits necessary to be considered in the binary expansions of such parameters such that with probability at least $99\%$ the difference $H(\tilde{\sigma}_G) - H(\sigma_G)$ represents a value smaller than $1\%$ of $R_H$. On each plot we compare the different methods developed in this paper, where the first method corresponds to the estimate from 
Theorem \ref{thmLinfty} and  the second method corresponds to the estimate from Theorem \ref{thmop}. In Figure \ref{fig:dig1}, we also included a third estimate from Example \ref{ex:completegraph} which is sharper and gives us better results when compared to the other methods. As we expected, the second method provides us with  better results when compared to the first one for complete graphs and for $N \times N$ King's graphs when $N$ is sufficiently large. On the other hand, for star graphs the first method is more appropriate, moreover, a certain  discrepancy of performance is easily observed. 


\section{Stability under a perturbed graph}\label{pertgraph}

In this section, we consider the stability of energy landscape when a given spin system defined on a graph is compressed into a smaller subsystem.
Differently from the previous sections, we fix a given Hamiltonian and we assume a sufficient condition that guarantees the existence of a nontrivial subset of the entire vertex set outside of which we can randomly assign any spin configuration and the energy of the system is kept under control up to a certain error margin. 

Let $G=(V,E)$ be a finite simple graph, and let $H$ be the Hamiltonian on $G$ given by
\begin{align*}
H(\sigma) = -\sum_{ \{x,y\} \in E} J_{x,y} \sigma_x\sigma_{y}
\end{align*}
for every configuration $\sigma\in\{-1,1\}^V$, where $\{J_{x,y}\}_{\{x,y\} \in E}$ is a collection of mutually independent standard Gaussian random variables.
What we would like to show is that we can compress the whole system into a nontrivial subsystem so that the energy landscape of such subsystem is close to the original one up to a given error margin. More precisely, our goal is to find a class of examples for which given a positive constant $\epsilon$, there is a positive $\delta$ such that the subsystem $V_0 = V_0(\delta)$ of $V$, defined from the relation
\begin{align*}
V\setminus V_0\coloneqq\left\{x\in V: \text{$\lvert J_{x,y}\rvert<\delta$ holds for every $y$ such that $\{x,y\} \in E$}\right\},
\end{align*}
is non-trivial, has size  comparable to the size of $V$, and satisfies
\begin{equation}\label{cond:ignore}
	\sup_{\sigma,\eta\in\{-1,1\}^V} \left\lvert H(\sigma)-H(\sigma_{V_0}, \eta_{V\setminus V_0}) \right\rvert < \vep R_H
\end{equation}
with high probability, see Figure \ref{fig:V0}.

\begin{figure}[h]
\centering
\includegraphics[scale=.2]{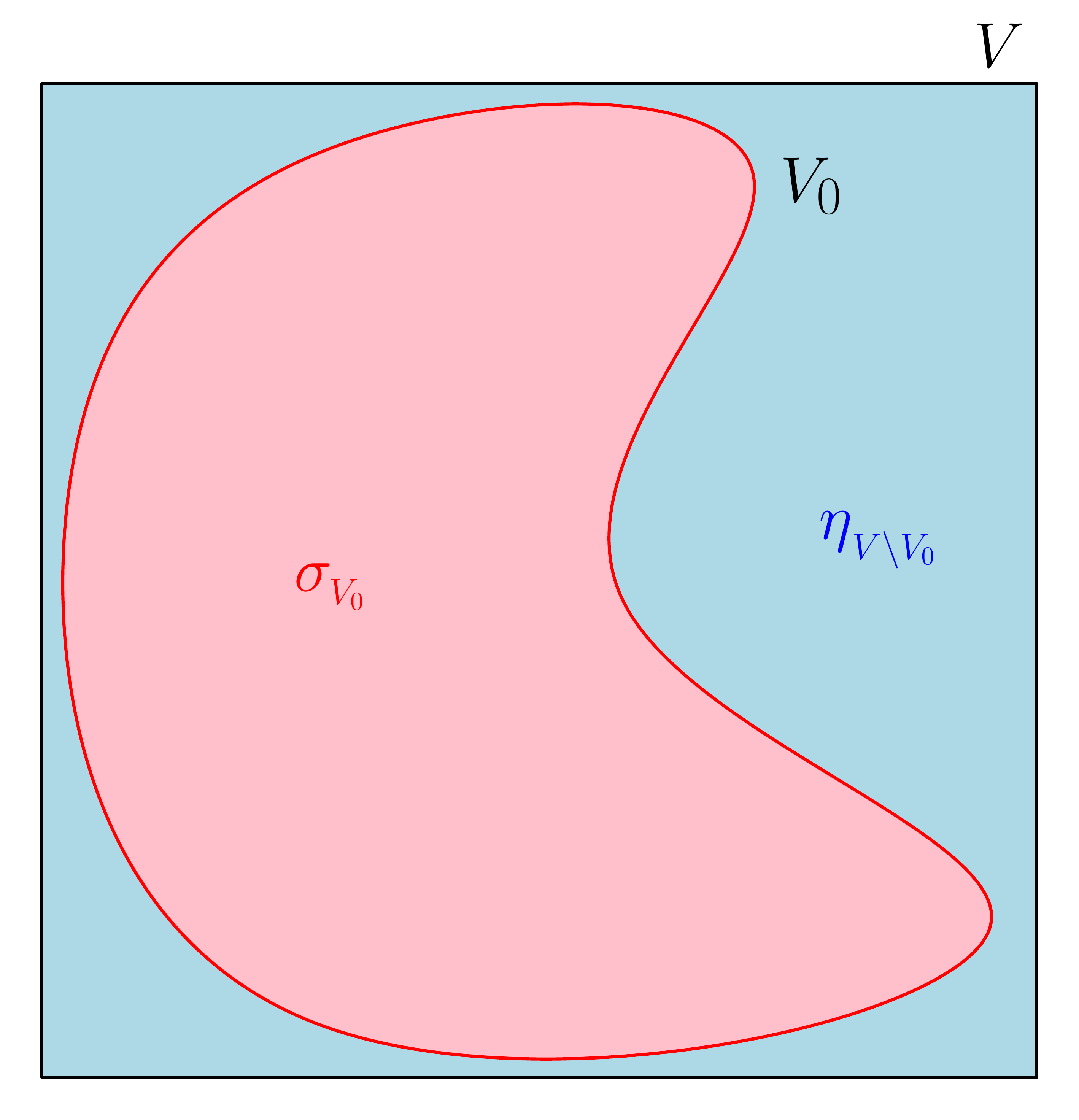}
\caption{We want to approximate the energy of a configuration $\sigma$ defined in the whole vertex set $V$ by the energy of a configuration that coincides with $\sigma$ in $V_{0}$ and
whose spins $\eta_{i}$'s in the set $V\setminus V_0$ are arbitrary.}\label{fig:V0}
\end{figure}

\subsection{One-dimensional discrete torus}\label{sec:1dtorus}

Let us solve the problem stated above in the particular case where the graph $G$ is a one-dimensional discrete torus.
\begin{theorem}\label{thm4}
Let $G=(V,E)$ be a one-dimensional discrete torus with $N$ vertices, that is, $V=\{1,2,\dots,N\}$ and $E=\{\{1,2\},\{2,3\},\dots,\{N-1,N\},\{N,1\}\}$. Given $\epsilon > 0$, let $\delta$ be a positive number such that 
$\delta < \epsilon/\sqrt{2\pi}$. Then, if $A$ is a subset of the event $\{0 < |V \backslash V_{0}| < N\}$, it follows that
\begin{equation}\label{thm4a}
\mP\left(\Bigg\{\sup_{\sigma,\tau\in\{-1,1\}^N} \left\lvert H(\sigma)-H(\sigma_{V_0}, \tau_{V\setminus V_0}) \right\rvert < \vep R_H \Bigg\}\cap A \right)
\geq \mP(A) - \frac{1 - \frac{2}{\pi}}{N \left(\sqrt{\frac{2}{\pi}} - \frac{2\delta}{\epsilon}\right)^{2}}
\end{equation}
holds for each $N \geq 3$. In particular, given constants $C > 0$ and $\alpha \in {[}0,1)$, we have
\begin{align}\label{thm4b}
\mP\left(\Bigg\{\sup_{\sigma,\tau\in\{-1,1\}^N} \left\lvert H(\sigma)-H(\sigma_{V_0}, \tau_{V\setminus V_0}) \right\rvert < \vep R_H \Bigg\} \cap \Big\{ C N^{\alpha} \leq |V \backslash V_{0}| < N \Big\}\right) &\nonumber \\
\geq \left( 1-\frac{C}{N^{1-\alpha}\theta^2} \right)^{2}\frac{1}{1+\frac{1+2\theta-3\theta^2}{N\theta^2}} - \theta^{N}  -& \frac{1 - \frac{2}{\pi}}{N \left(\sqrt{\frac{2}{\pi}} - \frac{2\delta}{\epsilon}\right)^{2}}
\end{align}
for $N$ sufficiently large, where 
\begin{equation}
\theta=\int_{-\delta}^{\delta}\frac{e^{-\xi^2/2}}{\sqrt{2\pi}}d\xi.
\end{equation}
\end{theorem}

Before we follow to the proof of the result above, let us clarify the theorem by providing the reader with some practical results. Let us consider the particular case where $C = 1$ and $\alpha \in (0,1)$. Corresponding to different values of $\epsilon$ and $\delta$, we obtain lower bounds for the probability that the size of $V \backslash V_0$ is at least $N ^ \alpha$ and condition (\ref{cond:ignore}) holds, see the table below.

\begin{table}[h!]
\centering
\begin{tabular}{ |p{1.5cm}||p{1.5cm}|p{1.5cm}|p{1.5cm}|p{2cm}|p{3.5cm}|}
 \hline
 \multicolumn{6}{|c|}{Examples} \\
 \hline
$N$ & $\epsilon$ & $\delta $ & $\alpha$& Minimum size of  $V \backslash V_0$  & Right-hand side of (\ref{thm4b})\\
 \hline
    $10^8$     & $0.05$ & $0.0198$  & $0.4$ & $1584$ & $0.877$ \\
  $10^8$     & $0.05$ & $0.0198$  & $0.5$ & $10^4$ & $0.361$ \\
 $10^8$     & $0.1$ & $0.0398$  & $0.5$ & $10^4$ & $0.810$ \\
 $10^{12}$ & $0.01$ &   $0.00398$  & $0.5$ & $10^ 6$ & $0.811$\\
 $10^{12}$  & $0.05 $&  $0.0199$ & $0.5$ &  $10^ 6$ & $0.992$\\
 $10^{12}$  & $0.1$ &  $0.0399$ & $0.5$&  $10^ 6$ & $0.998$\\
 $10^{12}$  & $0.05$ &  $0.0199$ & $0.6$&  $\approx 1.58 \times 10^7$ & $0.879$\\
 $10^{12}$  & $0.05$ &  $0.0199$ & $0.65$&  $\approx 6.31 \times 10^7 $ & $0.563$\\
  \hline
\end{tabular}
\caption{Applications of Theorem \ref{thm4}.}
\end{table}

Let us observe that for any pair $\sigma, \tau$ of spin configurations, we have
\begin{eqnarray*}
	|H(\sigma) - H(\sigma_{V_{0}}, \tau_{V \backslash V_{0}})| &=& \left|\sum_{x \in V_{0}} \sum_{\substack{y \in V \backslash V_{0} \\ \{x,y\} \in E}} J_{x,y} \sigma_{x}(\sigma_{y} - \tau_{y}) + 
	\sum_{\substack{\{x,y\} \subseteq V \backslash V_{0} \\ \{x,y\} \in E}} J_{x,y} (\sigma_{x} \sigma_{y} - \tau_{x} \tau_{y})\right| \\
	&=& \left|\sum_{x \in V_{0}} \sum_{\substack{y \in V \backslash V_{0} \\ \{x,y\} \in E}} J_{x,y} \sigma_{x}\sigma_{y}(1 - \sigma_{y}\tau_{y})  + 
	\sum_{\substack{\{x,y\} \subseteq V \backslash V_{0} \\ \{x,y\} \in E}} J_{x,y} \sigma_{x} \sigma_{y} (1 - \sigma_{x}\tau_{x} \sigma_{y}\tau_{y}) \right| 
\end{eqnarray*}
\begin{eqnarray*}
		&=& \left|\sum_{y \in V \backslash V_{0}} \sum_{\substack{x \in V \\ \{x,y\} \in E}} J_{x,y} \sigma_{x}\sigma_{y}\left[\mathbbm{1}_{x \in V_{0}} (1 - \sigma_{y}\tau_{y})  +
	\mathbbm{1}_{x \in V \backslash V_{0}} (1 - \sigma_{x}\tau_{x} \sigma_{y}\tau_{y})/2\right] \right|\\
	&\leq& 2 \sum_{y \in V \backslash V_{0}} \sum_{\substack{x \in V \\ \{x,y\} \in E}} |J_{x,y}| \leq 2\delta \sum_{y \in V \backslash V_{0}} \text{deg}(y).
\end{eqnarray*}
In particular, if $G$ is the one-dimensional torus as in Theorem \ref{thm4}, it follows that 
\begin{equation}\label{Hamiltoniandiff}
	\sup_{\sigma,\tau\in\{-1,1\}^N} \left\lvert H(\sigma)-H(\sigma_{V_0}, \tau_{V\setminus V_0}) \right\rvert \leq 4 \delta |V \backslash V_{0}|.
\end{equation}
Now, let us prepare two lemmas in order to prove Theorem \ref{thm4}.

\begin{lemma}\label{lem41}
For $R_H=\sup_{\xi,\eta}\lvert H(\xi)-H(\eta)\rvert$, we have
\begin{align}
2\sum_{x=1}^N\lvert J_{x,x+1}\rvert - 4\min_{x = 1,\dots, N}\lvert J_{x,x+1}\rvert
\le R_H
\le 2\sum_{x=1}^N\lvert J_{x,x+1}\rvert,
\end{align}
hence, with probability 1,
\begin{align*}
R_H \sim 2\sqrt{\frac{2}{\pi}}N \quad\text{as $N$ approaches infinity}.
\end{align*}
\end{lemma}

\begin{proof}
Without loss of generality, we assume $\min_x\lvert J_{x,x+1}\rvert=\lvert J_{N,1}\rvert$.
Let us fix $\sigma_{1}=1$.
Then, depending on the sign of $J_{1,2}$, we can determine $\sigma_2$ to minimize (or maximize) $H(\sigma)$.
We continue this procedure up to $\sigma_{N}$ and we have
\begin{align*}
\min_{\sigma\in\{-1,1\}^N}H(\sigma) &\leq -\sum_{x=1}^N\lvert J_{x,x+1}\rvert \ +2\min_{x = 1,\dots, N}\lvert J_{x,x+1}\rvert,\\
\max_{\sigma\in\{-1,1\}^N}H(\sigma) &\geq \sum_{x=1}^N\lvert J_{x,x+1}\rvert \ -2\min_{x = 1,\dots, N}\lvert J_{x,x+1}\rvert
\end{align*}
(the additional terms exist if frustration exist at $\sigma_N$ and $\sigma_1$).
Hence the inequality of the lemma is proven.

To show the last statement, we divide all terms by $N$ and use the law of large numbers for the folded normal distribution.
\end{proof}

\begin{lemma}\label{lem42}
If we assume  $N \geq 3$, then it follows that
\begin{align*}
\mathbb{E}\left[\lvert V\setminus V_0\rvert\right]=N\theta^2
\end{align*}
and
\begin{align*}
\mathbb{E}\left[\lvert V\setminus V_0\rvert^2\right]=N\theta^2\left(1+2\theta-3\theta^2+N\theta^2\right).
\end{align*}
\end{lemma}

\begin{proof}
For each $i=1,\dots,N$, let us define a random variable $X_{i}$ by letting
\begin{align*}
X_i=
\begin{cases}
1 & \text{if $\lvert J_{i-1,i}\rvert<\delta$ and $\lvert J_{i,i+1}\rvert<\delta$,}\\
0 & \text{otherwise}.
\end{cases}
\end{align*}
Then, by the definition of $V_0$, the condition $i\in V\setminus V_0$ is equivalent to $X_i = 1$.
Therefore, the expected value of the size of $V\setminus V_0$ will be given by
\begin{align*}
\mathbb{E}\left[\lvert V\setminus V_0\rvert\right]
= \mathbb{E}\left[\sum_{i=1}^N X_i \right]
= \sum_{i=1}^N\mathbb{E}\left[X_i\right]
= N\theta^2.
\end{align*}
Furthermore, we write
\begin{align*}
\lvert V\setminus V_0\rvert^2
=\sum_{i=1}^NX_i^2+2\sum_{i=1}^NX_iX_{i+1}+\sum_{i=1}^N\sum_{\substack{j\notin \{i-1,\, i,\, i+1\}}}X_iX_j.
\end{align*}
Here, the random variables $X_i$ and $X_{i+1}$ are not mutually independent but we have
\begin{align*}
X_iX_{i+1}=
\begin{cases}
1 & \text{if $\lvert J_{i-1,i}\rvert<\delta$,\ $\lvert J_{i,i+1}\rvert<\delta$ and $\lvert J_{i+1,i+2}\rvert<\delta$},\\
0 & \text{otherwise}.
\end{cases}
\end{align*}
Thus, it follows that the identity
\begin{align*}
\mathbb{E}\left[\lvert V\setminus V_0\rvert^2\right]
&=N\theta^2+2N\theta^3+N(N-3)\theta^4\\
&=N\theta^2\left(1+2\theta-3\theta^2+N\theta^2\right)
\end{align*}
holds, and we complete the proof.
\end{proof}

\begin{proof}[Proof of Theorem \ref{thm4}]
Let us start by splitting the probability in the left-hand side of equation (\ref{thm4a}) as 
\begin{align*}
\mP&\left(\Bigg\{\sup_{\sigma,\tau\in\{-1,1\}^N} \left\lvert H(\sigma)-H(\sigma_{V_0}, \tau_{V\setminus V_0}) \right\rvert < \vep R_H \Bigg\}\cap A \right)\\
&\geq \mP\left(\Bigg\{\sup_{\sigma,\tau\in\{-1,1\}^N} \left\lvert H(\sigma)-H(\sigma_{V_0}, \tau_{V\setminus V_0}) \right\rvert < \vep R_H \Bigg\} \cap \left\{\left| \frac{1}{N} \sum_{x = 1}^{N} |J_{x,x+1}| - \sqrt{\frac{2}{\pi}} \right| < \sqrt{\frac{2}{\pi}}  - \frac{2\delta}{\epsilon} \right\} \cap A \right)\\
& = \mP\left(\sup_{\sigma,\tau\in\{-1,1\}^N} \left\lvert H(\sigma)-H(\sigma_{V_0}, \tau_{V\setminus V_0}) \right\rvert < \vep R_H \,\middle\vert\, B \right)\,
\mP\left(B \right),
\end{align*}
where $B$ is the event given by
\begin{equation}
	B = \left\{\left| \frac{1}{N} \sum_{x = 1}^{N} |J_{x,x+1}| - \sqrt{\frac{2}{\pi}} \right| < \sqrt{\frac{2}{\pi}}  - \frac{2\delta}{\epsilon} \right\} \cap A. 
\end{equation}
From equation (\ref{Hamiltoniandiff}), Lemma \ref{lem41}, and the fact that, under condition $B$ (subset of $A$), $\min_{x = 1,\dots, N}|J_{x,x+1}| \leq \delta$, it follows that the conditional probability above satisfies
\begin{align*}
\mP\left(\sup_{\sigma,\tau\in\{-1,1\}^N} \left\lvert H(\sigma)-H(\sigma_{V_0}, \tau_{V\setminus V_0}) \right\rvert < \vep R_H \,\middle\vert\, B \right)
&\ge \mP\left( 2\delta\lvert V\setminus V_0\rvert <\vep\left( \sum_{x=1}^N\lvert J_{x,x+1}\rvert - 2\min_{x = 1,\dots, N}\lvert J_{x,x+1}\rvert\right)\,\middle\vert\, B  \right)\\
&\ge \mP\left( \frac{2\delta}{\vep} \lvert V\setminus V_0\rvert < \left( \sum_{x=1}^N\lvert J_{x,x+1}\rvert - 2\delta\right)\,\middle\vert\, B  \right)\\
&= \mP\left( \frac{2\delta}{\vep} \frac{\lvert V\setminus V_0\rvert + \epsilon}{N} < \frac{1}{N}\sum_{x=1}^N\lvert J_{x,x+1}\rvert \,\middle\vert\, B  \right)\\
&\ge \mP\left( \frac{2\delta}{\vep} < \frac{1}{N}\sum_{x=1}^N\lvert J_{x,x+1}\rvert \,\middle\vert\, B  \right) = 1,
\end{align*}
then
\begin{equation}\label{thm4:1}
	\mP\left(\sup_{\sigma,\tau\in\{-1,1\}^N} \left\lvert H(\sigma)-H(\sigma_{V_0}, \tau_{V\setminus V_0}) \right\rvert < \vep R_H \,\middle\vert\, B \right) = 1.
\end{equation}
The rest of the proof consists of estimating the probability of the event $B$. Let us write
\begin{equation}\label{thm4:2}
\mP(B)
\geq \mP(A)
+ \mP\left(\left| \frac{1}{N} \sum_{x = 1}^{N} |J_{x,x+1}| - \sqrt{\frac{2}{\pi}} \right| < \sqrt{\frac{2}{\pi}} - \frac{2\delta}{\epsilon} \right)
-1.
\end{equation}
It follows from Chebyshev's inequality that 
\begin{equation}\label{thm4:3}
	\mP\left(\left| \frac{1}{N} \sum_{x = 1}^{N} |J_{x,x+1}| - \sqrt{\frac{2}{\pi}} \right| \geq  \sqrt{\frac{2}{\pi}} - \frac{2\delta}{\epsilon}  \right)\leq \frac{\sigma_{FG}^{2}}{N \left(\sqrt{\frac{2}{\pi}} - \frac{2\delta}{\epsilon}\right)^{2}},
\end{equation}
where $\sigma_{FG}^{2}$ is the variance of the folded Gaussian random variable $Y = |J_{1,2}|$ which is equal to $1 - \frac{2}{\pi}$. By using equations (\ref{thm4:1}), (\ref{thm4:2}) and (\ref{thm4:3}), equation (\ref{thm4a})  follows.

In particular, if $A$ is the event given by $A = \{C N^{\alpha} \leq |V \backslash V_{0}| < N\}$. Note that
\begin{equation}\label{thm4:4}
\mP(A) = \mP(|V \backslash V_{0}| \geq C N^{\alpha}) - \mP(|V \backslash V_{0}| = N),
\end{equation}
where $\mP(|V \backslash V_{0}| = N) = \theta^{N}$. By the Paley-Zygmund inequality and Lemma \ref{lem42}, we have
\begin{align*}
\mP\left(\lvert V\setminus V_0\rvert \geq C N^{\alpha} \right)
&\ge { \left( 1 - \frac{C N^{\alpha}}{\mathbb{E}[\lvert V\setminus V_0\rvert]} \right)^2} \frac{\mathbb{E}[\lvert V\setminus V_0\rvert]^2}{\mathbb{E}[\lvert V\setminus V_0\rvert^2]}\\
&= { \left( 1 - \frac{C N^{\alpha}}{N \theta^{2}} \right)^2}\frac{1}{1+\frac{1+2\theta-3\theta^2}{N\theta^2}}
\end{align*}
for $N$ sufficiently large, therefore, equation (\ref{thm4b}) holds.
\end{proof}
 
\subsection{Generalizations}

The most natural step in further investigations is to extend the results obtained in Section \ref{sec:1dtorus} to the case where we include i.i.d. standard Gaussian external fields, and also extend such results to a larger class of examples such as to a $d$-dimensional torus or even to 
finite graphs with bounded degree. Note that, by assuming  the absence of external fields, in the same way as we obtained inequality (\ref{Hamiltoniandiff}), one can show that
\begin{equation}\label{Hamiltoniandiff2}
		|H(\sigma) - H(\sigma_{V_{0}}, \tau_{V \backslash V_{0}})| \leq  2\delta \sum_{y \in V \backslash V_{0}} \text{deg}(y)
\end{equation}
holds for any graph. So, analogously as in the one-dimensional torus case, it is expected that if we find a lower bound for $R_{H}$, as we did in Lemma \ref{lem41}, which is comparable to the right-hand side of equation (\ref{Hamiltoniandiff2}), then we may derive an extension of our results for a larger class of graphs. Some numerical results suggest that, for an Ising spin system in a $d$-dimensional torus with i.i.d. standard Gaussian spin-spin couplings and without external fields, $R_{H}$ is still of order $N$, but it still lacks a rigorous proof of that observation  due to the difficulty of dealing with frustrated configurations in a higher dimensional torus.

The simulations presented in this section were performed by using a modified version of the stochastic cellular automata algorithm studied in \cite{HKKS19,HKKS21} to estimate the maximum and minimum value 
of the Hamiltonian $H$ in order to find an approximation of  $R_{H}$ corresponding to different values of $N$. Note that, in such plots, each dot represents the value of $R_{H}$ (resp. $R_{H}/N$) corresponding to a torus with $N$ vertices for a realization of the random values of spin-spin couplings (i.i.d. standard Gaussian random variables). In the one-dimensional case (see Figure \ref{fig:torus1d}), we see that the value of $R_{H}/N$ approximates the value $2 \sqrt{2/\pi} \approx 1.5957$, as expected due to Lemma \ref{lem41}.
\begin{figure}[H]
        \centering
        \begin{subfigure}[b]{0.44\textwidth}
                \includegraphics[width=\textwidth]{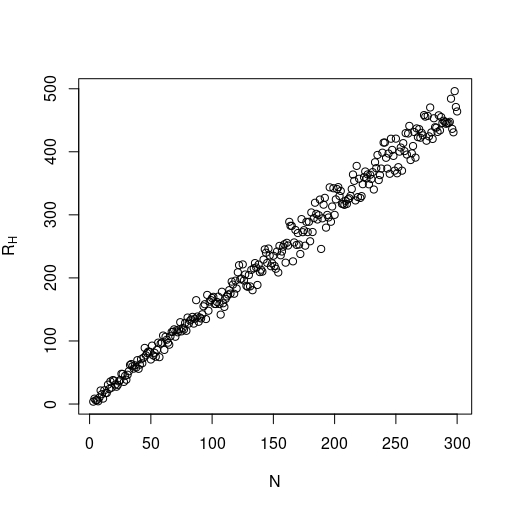}
                \caption{}
                \label{fig:lab0}
        \end{subfigure}       
        \begin{subfigure}[b]{0.44\textwidth}
                \includegraphics[width=\textwidth]{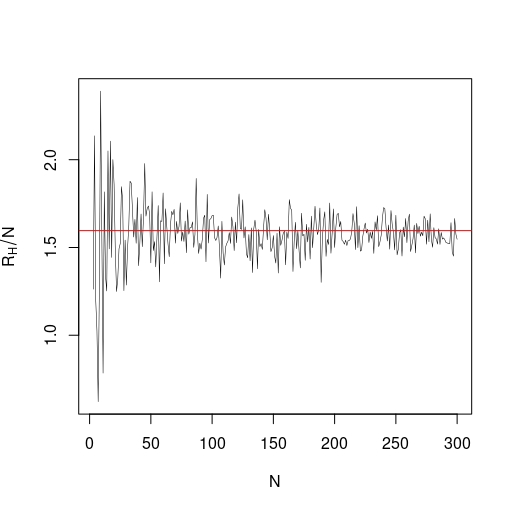}
                \caption{}
                \label{fig:lab1}
        \end{subfigure}
    \caption{The dependence of $R_{H}$ with respect to the size of the system $N$ in the one dimensional case and its asymptotic behavior as $N$ grows.}
    \label{fig:torus1d}
\end{figure}

Now, for the two and three dimensional cases (see Figure \ref{fig:torus}), when we consider larger values of $N$, the value of  $R_{H}/N$ seems to approximate the values $2.564$ and $3.329$, respectively. Note that such simulated values represent lower bounds for the real value of the limit $R_{H}/N$ as $N$ approaches infinity, so the true limits are still unknown. Furthermore, we conjecture that such limit exists in any dimension and the random variable $R_{H}/N$ converges almost surely due to the fact that, in higher dimension, its simulated values seem to fluctuate less around an asymptotic limit as compared to the one-dimensional case.

It is straightforward to show that, for the $d$-dimensional torus, we have
\begin{equation*}
R_H \leq 2 \sum_{k = 1}^{d}\sum_{i \in V} |J_{i, i+\mathbf{e_k}}|,
\end{equation*}
where $\mathbf{e_k}$ stands for the $k$-th canonical vector of the $d$-dimensional Euclidean space, then
\begin{equation}
\limsup_N \frac{R_H}{N} \leq 2d \sqrt{\frac{2}{\pi}}.
\end{equation}
Moreover, it follows from the fact that $R_H \geq \sqrt{\sum_b J_b^2}$ (see Lemma \ref{lem:total margin}) and the Cauchy–Schwarz inequality that 
\begin{equation}
\frac{1}{\sqrt{N}} \sum_b |J_b| \leq R_H.
\end{equation}
Therefore, we see that there is still room for improvement and the need of rigorous proofs about the existence and determination of the limit $\lim_{N \to \infty}R_H/N$, originating a mathematical problem which is interesting by itself.
   
   \begin{figure}[H]
        \centering
        \begin{subfigure}[b]{0.44\textwidth}
                \includegraphics[width=\textwidth]{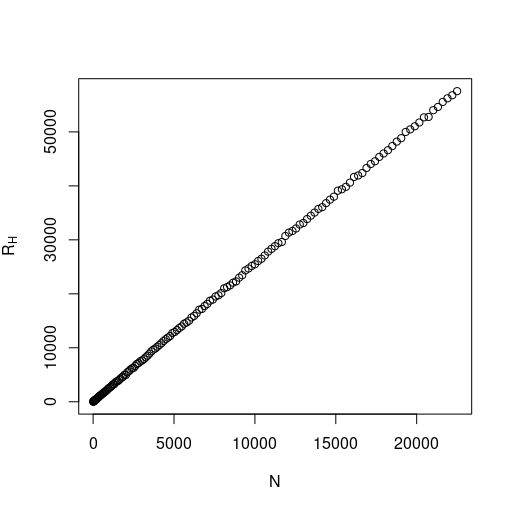}
                \caption{Two dimensional.}
                \label{fig:lab2}
        \end{subfigure}       
        \begin{subfigure}[b]{0.44\textwidth}
                \includegraphics[width=\textwidth]{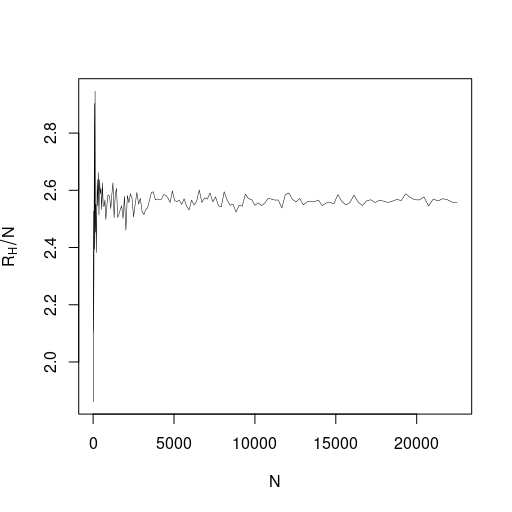}
                \caption{Two dimensional.}
                \label{fig:lab3}
        \end{subfigure}
        \begin{subfigure}[b]{0.44\textwidth}
                \includegraphics[width=\textwidth]{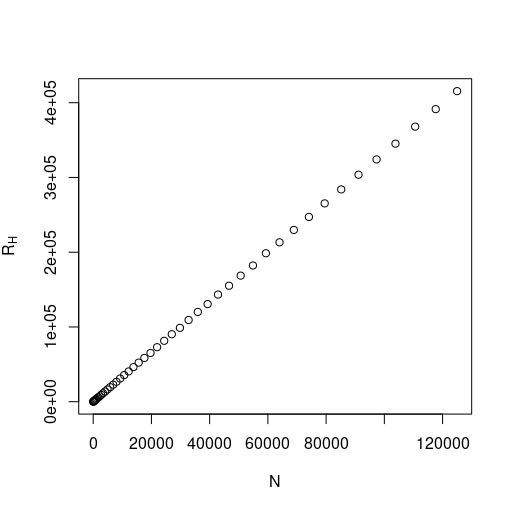}
                \caption{Three dimensional.}
                \label{fig:lab4}
        \end{subfigure}       
        \begin{subfigure}[b]{0.44\textwidth}
                \includegraphics[width=\textwidth]{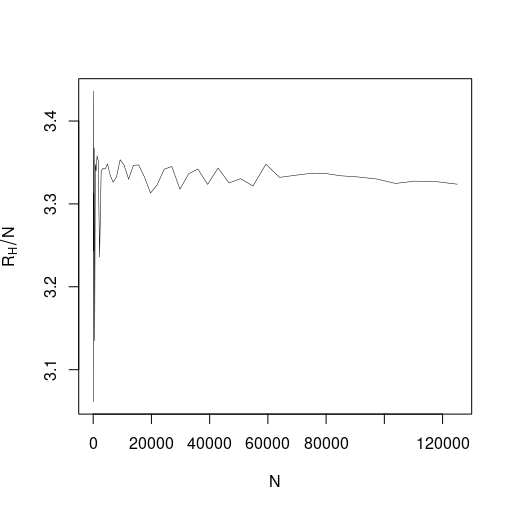}
                \caption{Three dimensional}
                \label{fig:lab5}
        \end{subfigure}
    \caption{The dependence of $R_{H}$ with respect to the size of the system $N$ in the two dimensional and three dimensional cases and their asymptotic behavior as $N$ grows.}\label{fig:torus}
\end{figure}
 
\subsection*{Acknowledgment}
This work was supported by JST CREST Grant Number JP22180021, Japan. We would like to thank Takashi Takemoto and Normann Mertig of Hitachi, Ltd., for providing us with a stimulating platform for the weekly meeting at Global Research Center for Food \& Medical Innovation (FMI) of Hokkaido University. We would also like to thank Hiroshi Teramoto of Kansai University, as well as  Masamitsu Aoki, Yoshinori Kamijima, Katsuhiro Kamakura, Suguru Ishibashi and Takuka Saito of Mathematics Department, for valuable comments and encouragement at the aforementioned meetings at FMI.


\end{document}